\documentclass{cccg20}
\usepackage{graphicx,amssymb,amsmath}
\usepackage{hyperref,aliascnt,cite}
\usepackage{balance}
\DeclareUrlCommand\path{\urlstyle{same}}

\newtheorem{theorem}{Theorem}[section]

\newaliascnt{lemma}{theorem} \newtheorem{lemma}[lemma]{Lemma}
\aliascntresetthe{lemma} 

\newaliascnt{corollary}{theorem}

\aliascntresetthe{corollary}

\def\emph#1{\textit{\textbf{\boldmath #1}}}

\title{Acutely Triangulated, Stacked, and Very Ununfoldable Polyhedra}

\author{%
  Erik D. Demaine%
    \thanks{Computer Science and Artificial Intelligence Laboratory,
      Massachusetts Institute of Technology,
      \protect\url{{edemaine,mdemaine}@mit.edu}}
\and
  Martin L. Demaine\footnotemark[1]
\and
  David Eppstein\thanks{Computer Science Department,
        University of California, Irvine, \protect\url{eppstein@uci.edu}. This work was supported in part by the US National Science Foundation under grant CCF-1616248.}
}

\index{Demaine, Erik D.}
\index{Demaine, Martin L.}
\index{Eppstein, David}

\begin{document}
\thispagestyle{empty}
\maketitle

\begin{abstract}
We present new examples of topologically convex edge-ununfoldable polyhedra,
i.e., polyhedra that are combinatorially equivalent to convex polyhedra,
yet cannot be cut along their edges and unfolded into one planar piece
without overlap.
One family of examples is \emph{acutely triangulated},
i.e., every face is an acute triangle.
Another family of examples is \emph{stacked},
i.e., the result of face-to-face gluings of tetrahedra.
Both families achieve another natural property,
which we call \emph{very ununfoldable}:
for every $k$, there is an example such that every nonoverlapping
multipiece edge unfolding has at least $k$ pieces.

\end{abstract}

\section{Introduction}

Can every convex polyhedron be cut along its edges and unfolded into a
single planar piece without overlap?
Such \emph{edge unfoldings} or \emph{nets} are useful for constructing 3D
models of a polyhedron (from paper or other material such as sheet metal):
cut out the net, fold along the polyhedron's uncut edges,
and re-attach the polyhedron's cut edges \cite{Wen-PM-71}.
Unfoldings have also proved useful in computational geometry algorithms
for finding shortest paths on the surface of polyhedra
\cite{AgaAroORo-SICOMP-97,AroORo-DCG-92,CheHan-SoCG-90}.

Edge unfoldings were first described in the early 16th century by
Albrecht D\"urer \cite{Duerer-1525}, implicitly raising the still-open
question of whether every convex polyhedron has one
(sometimes called D\"urer's conjecture).
The question was first formally stated in 1975 by G. C. Shephard,
although without reference to D\"urer~\cite{Fri-HFM-18,She-MPCPS-75}.
It has been heavily studied since then, with progress of two types
\cite{Demaine-O'Rourke-2007,ORo-19}:
\begin{enumerate}
\item finding restricted classes of polyhedra, or generalized types of unfoldings, for which the existence of an unfolding can be guaranteed; and
\item finding generalized classes of polyhedra, or restricted types of unfoldings, for which counterexamples --- \emph{ununfoldable polyhedra} --- can be shown to exist.
\end{enumerate}

Results guaranteeing the existence of an unfolding include:
\begin{itemize}
\item Every pyramid, prism, prismoid, and dome has an edge unfolding
  \cite{Demaine-O'Rourke-2007}.
\item Every sufficiently flat acutely triangulated convex terrain has an edge unfolding~\cite{caps}. Consequentially, every acutely triangulated convex polyhedron can be unfolded into a number of planar pieces that is bounded in terms of the ``acuteness gap'' of the polyhedron, the minimum distance of its angles from a right angle.
\item Every convex polyhedron has an affine transformation
  that admits an edge unfolding \cite{Ghomi-2014}.
\item Every convex polyhedron can be unfolded to a single planar piece by cuts interior to its faces
\cite{AgaAroORo-SICOMP-97,SunUnfolding_EGC2011f}.
\item Every polyhedron with axis-parallel sides can be unfolded after a linear number of axis-parallel cuts through its faces
\cite{Genus2Unfolding_GC}.
\item Every triangulated surface (regardless of genus) has a ``vertex unfolding'', a planar layout of triangles connected through their vertices that can be folded into the given surface~\cite{DemEppEri-DG-03}.
\item For ideal polyhedra in hyperbolic space, unlike Euclidean convex polyhedra or non-ideal hyperbolic polyhedra, every spanning tree forms the system of cuts of a convex unfolding into the hyperbolic plane.
\end{itemize}

Previous constructions of ununfoldable polyhedra include the following results.
A polyhedron is \emph{topologically convex} if it is combinatorially
equivalent to a convex polyhedron, meaning that its surface is a topological
sphere and its graph is a 3-vertex-connected planar graph.
\begin{itemize}
\item Some orthogonal polyhedra and topologically convex orthogonal polyhedra have no edge unfolding, and it is NP-complete to determine whether an edge unfolding exists in this case~\cite{BieDemDem-CCCG-98,AbeDem-CCCG-11}.
\item There exists a convex-face star-shaped topologically convex polyhedron with no edge unfolding~\cite{Tar-UMN-99,Gruenbaum-2002-net}.
\item There exists a triangular-face topologically convex polyhedron with no edge unfolding~\cite{BerDemEpp-CGTA-03}.
\item There exist edge-ununfoldable topologically convex polyhedra with as few as 7 vertices and 6 faces, or 6 vertices and 7 faces~\cite{AkiDemEpp-JCDCG-19}.
\item There exists a topologically convex polyhedron that does not even have a vertex unfolding~\cite{AbeDemDem-CCCG-11}.
\item There exist domes that have no \emph{Hamiltonian unfolding}, in which the cuts form a Hamiltonian path through the graph of the polyhedron~\cite{DemDemUeh-CCCG-13}. Similarly, there exist polycubes that have no Hamiltonian unfolding~\cite{DemDemEpp-19}.
\item There exists a convex polyhedron, equipped with 3-vertex-connected planar graph of geodesics partitioning the surface into regions metrically equivalent to convex polygons, that cannot be cut and unfolded along graph edges~\cite{BarGho-DCG-19}.
\end{itemize}


In this paper, we consider two questions left open by the previous work on
edge-ununfoldable polyhedra with triangular faces~\cite{BerDemEpp-CGTA-03},
and strongly motivated by O'Rourke's recent results on unfoldings of acutely-triangulated polyhedra~\cite{caps}.
First, the previous counterexample of this type involved triangles with highly
obtuse angles. Is this a necessary feature of the construction, or does there
exist an ununfoldable polyhedron with triangular faces that are all acute?
Second, how far from being unfoldable can these examples be? Is it possible to
cut the surfaces of these polyhedra into a bounded number of planar pieces
(instead of a single piece) that can be folded and glued to form the
polyhedral surface?
(Both questions are motivated by previously posed analogous questions for
convex polyhedra, as easier versions of D\"urer's conjecture
\cite[Open Problems 22.12 and 22.17]{Demaine-O'Rourke-2007}.)

We answer both of these questions negatively, by finding families of
topologically convex edge-ununfoldable polyhedra with all faces acute
triangles, in which any cutting of the surface into regions that can be
unfolded to planar pieces must use an arbitrarily large number of pieces.
Additionally, we use a similar construction to prove that there exist
edge-ununfoldable \emph{stacked polyhedra} \cite{Gru-Geomb-01b},
formed by gluing tetrahedra face-to-face with the gluing pattern of a tree,
that also require an arbitrarily large number of pieces to unfold.
We leave open the question of whether there exists an edge-ununfoldable
stacked polyhedron with acute-triangle faces.

\section{Hats}

\begin{figure}[t]
\includegraphics[width=\columnwidth]{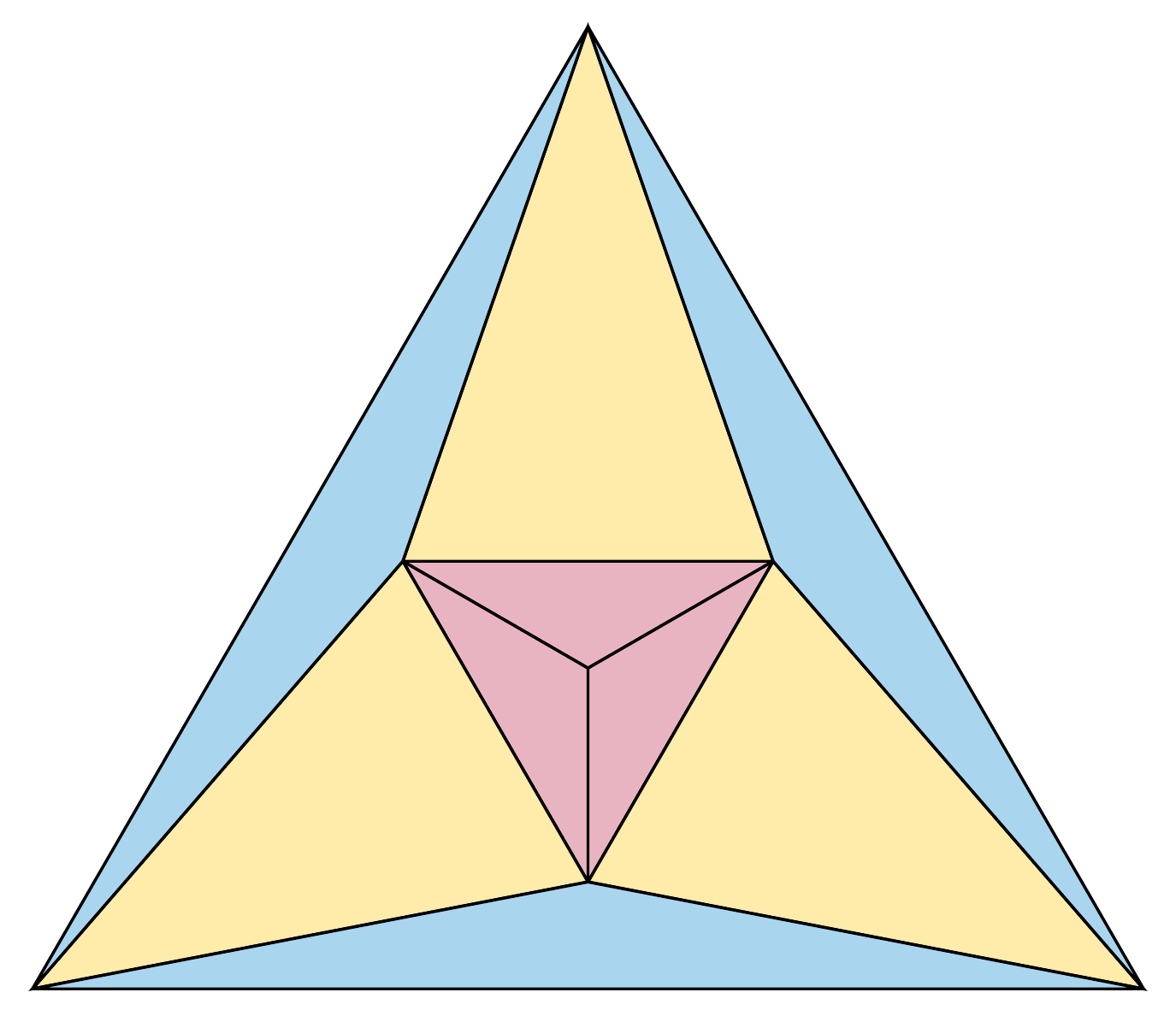}
\caption{Combinatorial structure of a hat}
\label{fig:hat}
\end{figure}

Our construction follows that of Bern et al.~\cite{BerDemEpp-CGTA-03} in being based on certain triangulated topological disks, which they called \emph{hats}. The combinatorial structure of a hat (in top view, but with different face angles than the hat we use in our proof) is shown in \autoref{fig:hat}: It consists of nine triangles, three of which (the \emph{brim}, blue in the figure) have one edge on the outer boundary of the disk. The next three triangles, yellow in the figure, have a vertex but not an edge on the disk boundary; we call these the \emph{band} of the hat. The central three triangles, pink in the figure, are disjoint from the boundary and meet at a central vertex; we call these the \emph{crown} of the hat.

In both the construction of Bern et al.~\cite{BerDemEpp-CGTA-03} and in our construction, the three vertices of the hat that are interior to the disk but not at the center all have negative curvature, meaning that the sum of the angles of the faces meeting at these vertices is greater than $2\pi$. The center vertex, on the other hand, has positive curvature, a sum of angles less than $2\pi$. When this happens, we can apply the following lemmas:

\begin{lemma}
\label{lem:cut-2-neg}
At any negatively-curved vertex of a polyhedron, any unfolding of the polyhedron that cuts only its edges and separates its surface into one or more simple polygons must cut at least two edges at each negatively-curved vertex.
\end{lemma}

\begin{proof}
If only one edge were cut then the faces surrounding that vertex could not unfold into the plane without overlap.
\end{proof}

\begin{lemma}
\label{lem:tree-1-boundary}
Let $D$ be a subset of the faces of a polyhedron, such that the polyhedron is topologically a sphere and $D$ is topologically equivalent to a disk (such as a hat).
Then in any unfolding of the polyhedron (possible cutting it into multiple pieces), either $D$ is separated into multiple pieces by a path of cut edges from one boundary vertex of $D$ to another or by a cycle of cut edges within $D$, or the set of cut edges within $D$ forms a forest with at most one boundary vertex for each tree in the forest.
\end{lemma}

\begin{proof}
If the cut edges within $D$ do not form a forest, they contain a cycle and the Jordan Curve Theorem implies that this cycle separates an interior part of the boundary from the exterior. If they form a forest in which some tree contains two boundary vertices, then they contain a boundary-to-boundary path within $D$, again separating $D$ by the Jordan Curve Theorem. The only remaining possibility is a forest with at most one boundary vertex per tree.
\end{proof}

\begin{lemma}
\label{lem:hat-path}
For a hat combinatorially equivalent to the one in \autoref{fig:hat}, with positive curvature at the center vertex and negative curvature at the other three interior vertices, any unfolding that does not cut the hat into multiple pieces must cut a set of edges along a single path from a boundary vertex to the center vertex.
\end{lemma}

\begin{proof}
By \autoref{lem:tree-1-boundary}, each component of cut edges must form a tree with at most one boundary vertex within the hat. But every tree with one or more edges has at least two leaves, and every tree that is not a path has at least three leaves. By \autoref{lem:cut-2-neg}, the only non-boundary leaf can be the center vertex, so each component must be a path from the boundary to this vertex.
\end{proof}

Up to symmetries of the hat, there are only two distinct shapes that the path of \autoref{lem:hat-path} from the boundary to the center of a hat can have (\autoref{fig:hatpaths}). These two cuttings differ in how the crown triangles are attached to the band and to each other, but they both cut the brim and band triangles in the same way, into a strip of triangles connected edge-to-edge around the boundary of the hat.

\begin{figure}[t]
\includegraphics[width=\columnwidth]{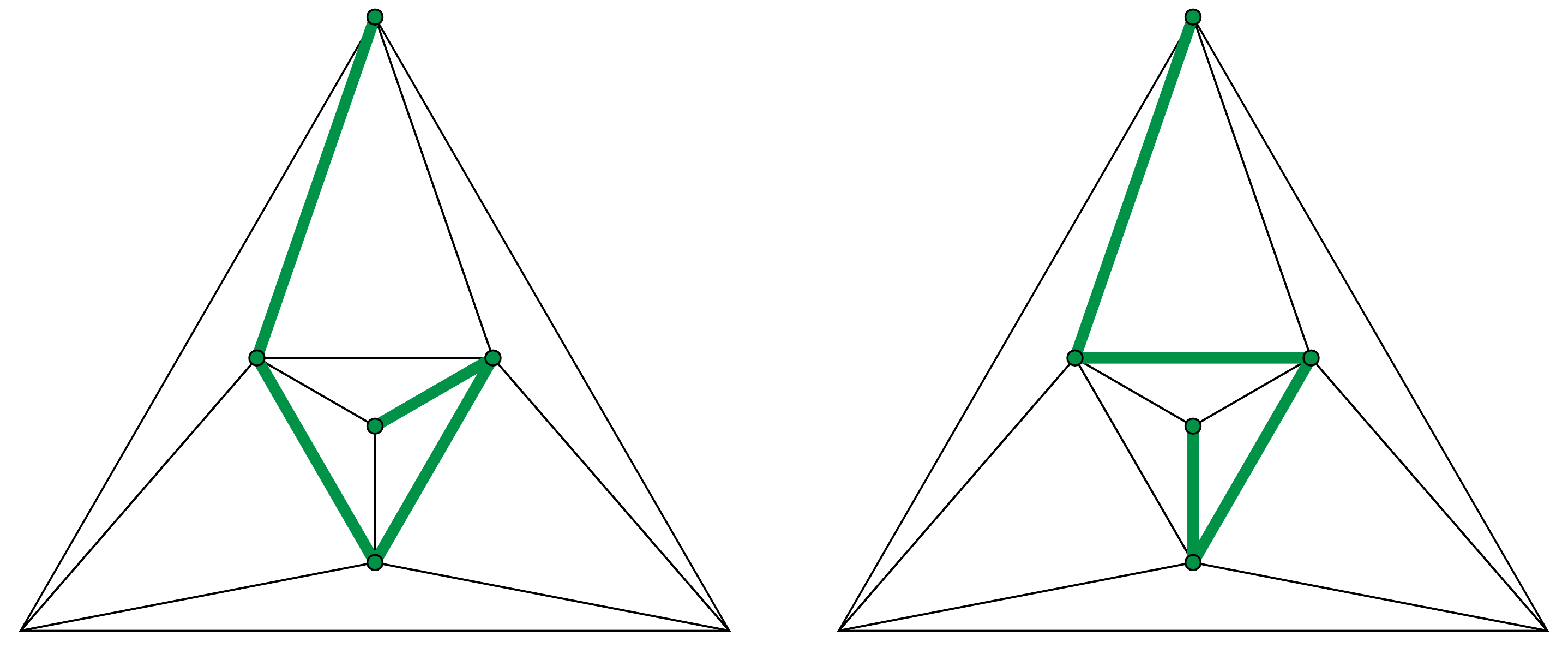}
\caption{Two paths from a boundary vertex of a hat, through all three negatively curved vertices, to the center vertex}
\label{fig:hatpaths}
\end{figure}

Our key new construction is depicted in unfolded (but self-overlapping) form in \autoref{fig:spiked-hat}. It is a hat in which all triangles are acute and isosceles:
\begin{itemize}
\item The three brim triangles have apex angle $85^\circ$ and base angle $47.5^\circ$.
\item The three band triangles have base angle $85^\circ$ and apex angle $10^\circ$.
\item The three crown triangles are congruent to the band triangles, with base angle $85^\circ$ and apex angle $10^\circ$.
\end{itemize}
As in the construction of Bern et al.~\cite{BerDemEpp-CGTA-03}, this leaves negative curvature (total angle $425^\circ$ from five $85^\circ$ angles) at the three non-central interior angles of the hat, and positive curvature (total angle $30^\circ$) at the center vertex, allowing the lemmas above to apply. The cut edges of the figure form a tree with a degree-three vertex at one of the negatively curved vertices of the hat, and a leaf at another negatively curved vertex, the one at which the self-overlap of the figure occurs, So the cutting in the figure does not match in detail either of the two path cuttings of \autoref{fig:hatpaths}. Nevertheless, the brim and band triangles are unfolded as they would be for either of these two path cuttings. It is evident from the figure that this unfolding of the brim and band triangles cannot be extended to a one-piece unfolding of the entire hat: if a crown triangle is attached to the middle of the three unfolded band triangles (as it is in the figure) then there is no room on either side of it to attach the other two crown triangles, and a crown triangle attached to either of the other two band triangles would overlap the opposite band triangle. We prove this visual observation more formally below.

\begin{figure}[t]
\includegraphics[width=\columnwidth]{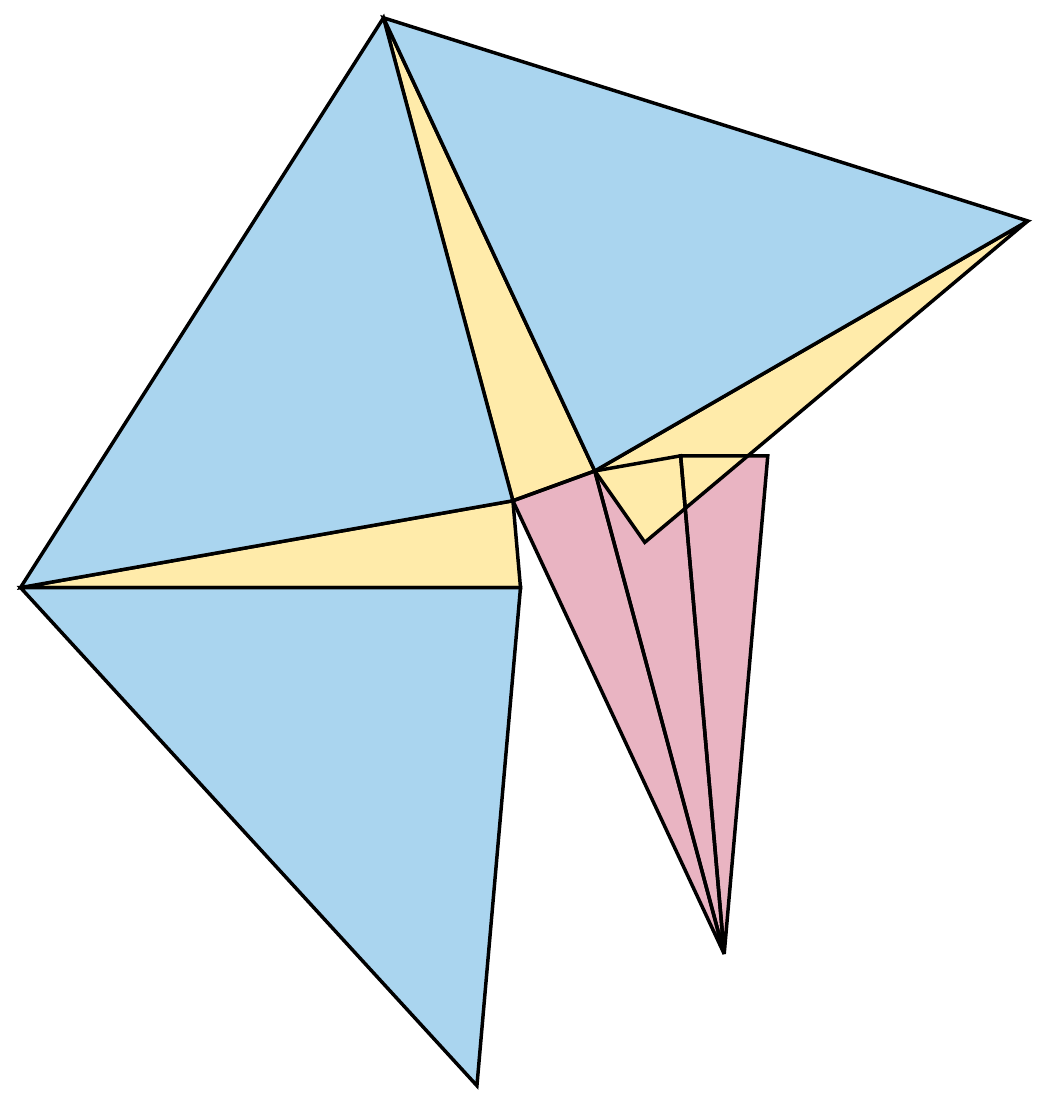}
\caption{A hat made with acute isosceles triangles. Unlike \autoref{fig:hatpaths}, the cuts made to form the self-overlapping unfolding shown do not form a path.}
\label{fig:spiked-hat}
\end{figure}

\begin{lemma}
\label{lem:hat}
The hat with acute triangles described above has no single-piece unfolding.
\end{lemma}

\begin{proof}
As we have already seen in \autoref{lem:hat-path}, any unfolding (if it exists) must be along one of the two cut paths depicted in \autoref{fig:hatpaths}. As a result, the unfolding of the brim and band triangles (but not the crown triangles) must be as depicted in \autoref{fig:spiked-hat}. In this unfolding, the three base sides of the unfolded band triangles form a polygonal chain whose interior angles (surrounding the central region of the figure where the pink crown triangles are attached) can be calculated as $105^\circ$.

\begin{figure}[t]
\centering
\includegraphics[width=0.6\columnwidth]{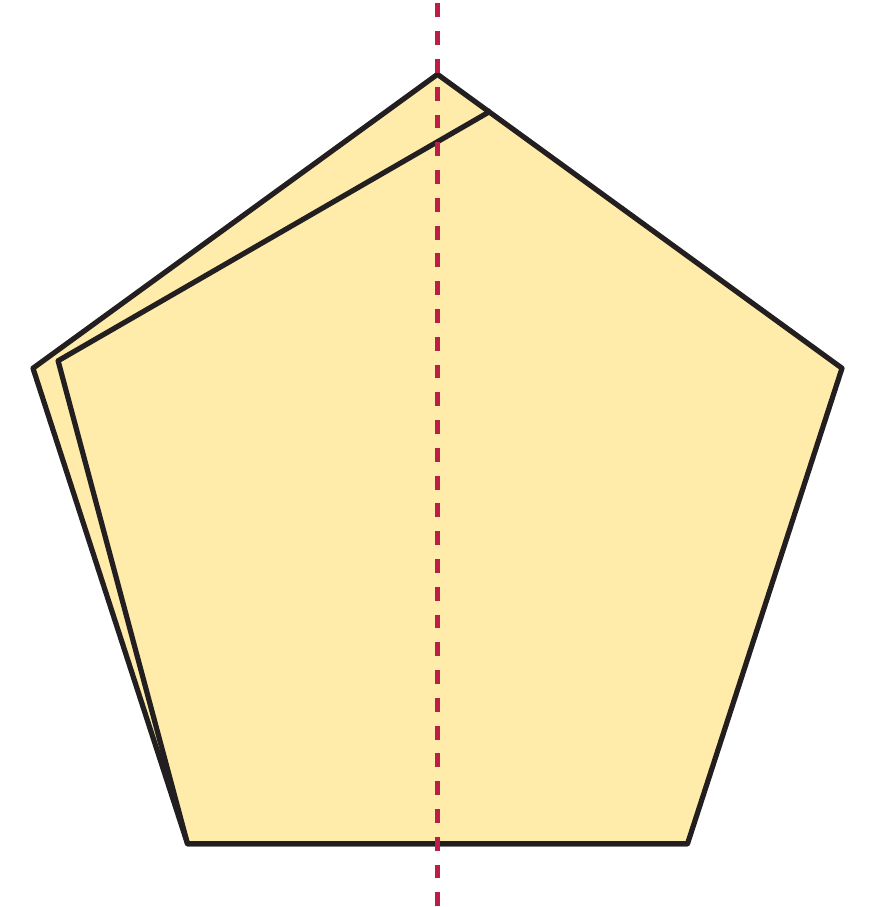}
\caption{Each vertex of a regular pentagon lies on the perpendicular bisector of the opposite side; in a path of three equal edges with the tighter angle $105^\circ$, the last edge overlaps the perpendicular bisector of the first.}
\label{fig:midline}
\end{figure}

A regular pentagon has interior angles of $108^\circ$, and has the property that each vertex lies on the perpendicular bisector of the opposite edge. Because the interior angles of the chain of base sides of band triangles are $105^\circ$, less than this $108^\circ$ angle, it follows that the band triangle at one end of the chain extends across the perpendicular bisector of the base edge at the other end of the chain. Further, it does so at a point closer than the vertex of a regular pentagon sharing this same base edge (\autoref{fig:midline}).

If a crown triangle were attached to one of the two base edges at the ends of the chain of three base edges, its altitude would lie along the perpendicular bisector of the base edge. And because the crown triangle has an apex angle of $10^\circ$, sharper than the angle of an isosceles triangle inscribed within a regular pentagon, its altitude extends across the perpendicular bisector farther than the regular pentagon vertex, causing it to overlap with the band triangle at the other end of the chain of three base edges.

Therefore, attaching a crown triangle to either the first or last of the band triangle base edges in the chain of these three edges necessarily leads to a self-overlapping unfolding. However, these two ways of attaching a crown triangle are the only ones permitted by the two cases depicted in \autoref{fig:hatpaths}. Attaching a crown triangle to the middle of the three base edges, as in \autoref{fig:spiked-hat}, can only be done by cutting along a tree that is not a path. Therefore, no unfolding exists.
\end{proof}

The following construction is straightforward, and will allow us to construct polyhedra with multiple hats while keeping the hats disjoint from each other. 

\begin{lemma}
The hat with acute triangles described above can be realized in three-dimensional space,
lying within a right equilateral-triangle prism whose base is the boundary of the hat.
\end{lemma}

\section{Acute Ununfoldable Polyhedra}

We now use these hats to construct a topologically convex ununfoldable polyhedron.

\begin{figure}[t]
\includegraphics[width=\columnwidth]{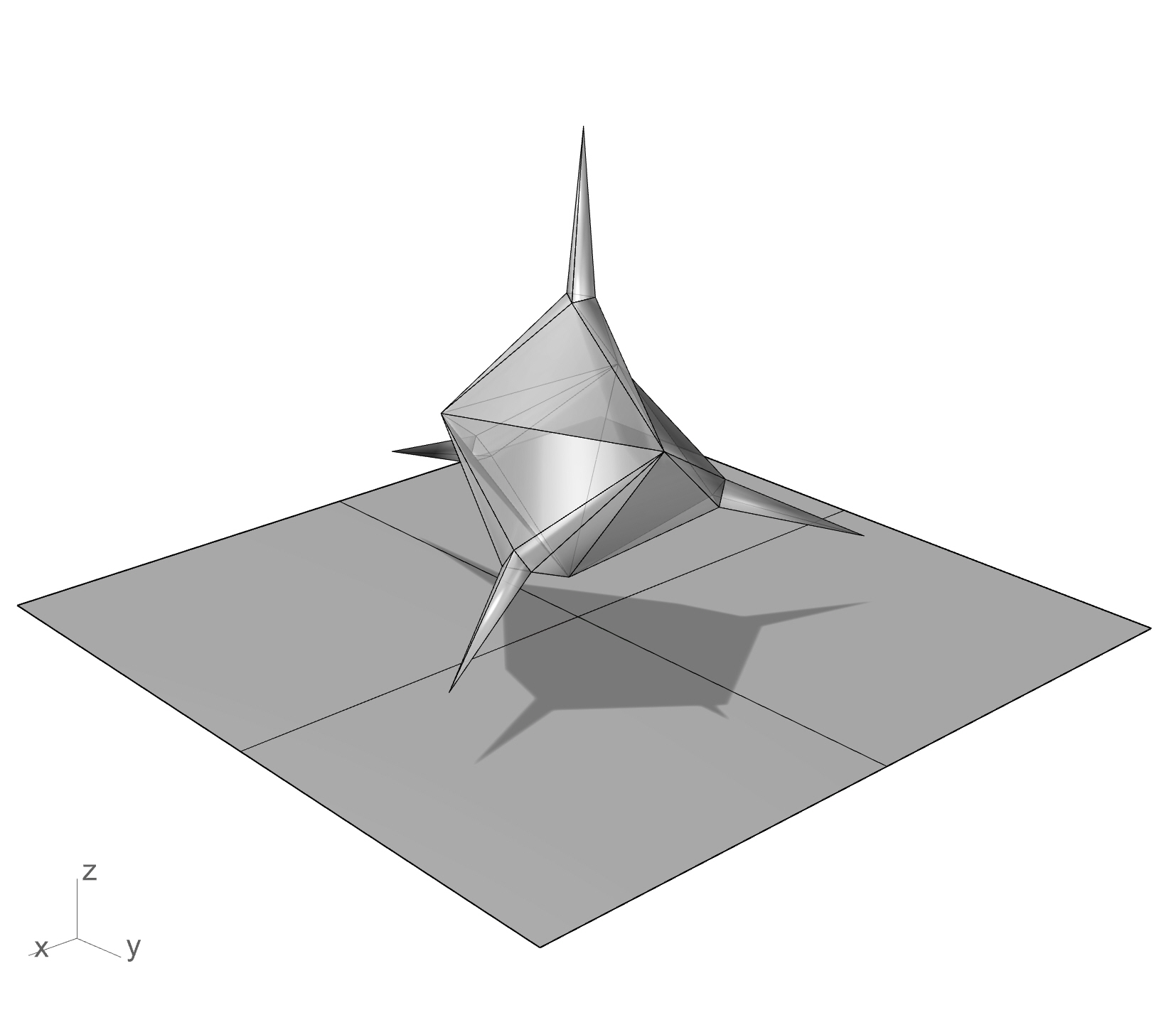}
\caption{Tetrahedron with faces replaced by hats}
\label{fig:caltrop}
\end{figure}

\begin{theorem}
\label{thm:ununf}
There exists a topologically convex ununfoldable polyhedron whose faces are all acute isosceles triangles.
\end{theorem}

\begin{proof}
Replace the four faces of a regular tetrahedron by acute-triangle hats, all pointing outward, as shown in \autoref{fig:caltrop}.
Because each lies within a prism having the tetrahedron face as a base,
they do not overlap each other in space.
By \autoref{lem:hat}, no hat can be unfolded into a single piece, so any possible unfolding (even one into multiple pieces) must cut each hat along some path between two of its three boundary vertices (at least; there may be more cuts besides these).
The four paths formed in this way are disjoint except at their ends,
and connect the four vertices of the tetrahedron,
necessarily forming at least one cycle
that separates the tetrahedron into at least two pieces.
\end{proof}

Like the examples of Tarasov, Gr\"unbaum, and Bern et al.~\cite{Tar-UMN-99,Gruenbaum-2002-net,BerDemEpp-CGTA-03}, the resulting polyhedron is also star-shaped, with the center of the tetrahedron in its kernel.

\section{Stacked Ununfoldable Polyhedra}
\label{sec:stacked}

A \emph{stacked polyhedron} is a polyhedron that can be formed by repeatedly gluing a tetrahedron onto a single triangular face of a simpler stacked polyhedron, starting from a single tetrahedron~\cite{Gru-Geomb-01b}. To make ununfoldable stacked polyhedra, we use a similar strategy to our construction of ununfoldable polyhedra with acute-triangle-faces, in which we replace some faces of a convex polyhedron by hats. However, the acute-triangle hat that we used earlier cannot be used as part of a stacked polyhedron: in a stacked polyhedron, every non-face triangle is subdivided into three smaller triangles, but that is not true of the outer triangle of \autoref{fig:hat}. Instead, we use the hat shown in \autoref{fig:stacked-hat}.
As before, it has three brim triangles, three band triangles, and three crown triangles, but they are arranged differently and less symmetrically.  We make the brim and band triangles nearly coplanar, with shapes approximating those shown in the figure, but projecting slightly out of the figure so that the result can be constructed as a stacked polyhedron. We choose the crown triangles to be isosceles, and taller than the isosceles triangles inscribed in regular pentagons, as in our acute-triangle construction, so that (as viewed in  \autoref{fig:stacked-hat}) they project out of the figure.

\begin{figure}[t]
\includegraphics[width=\columnwidth]{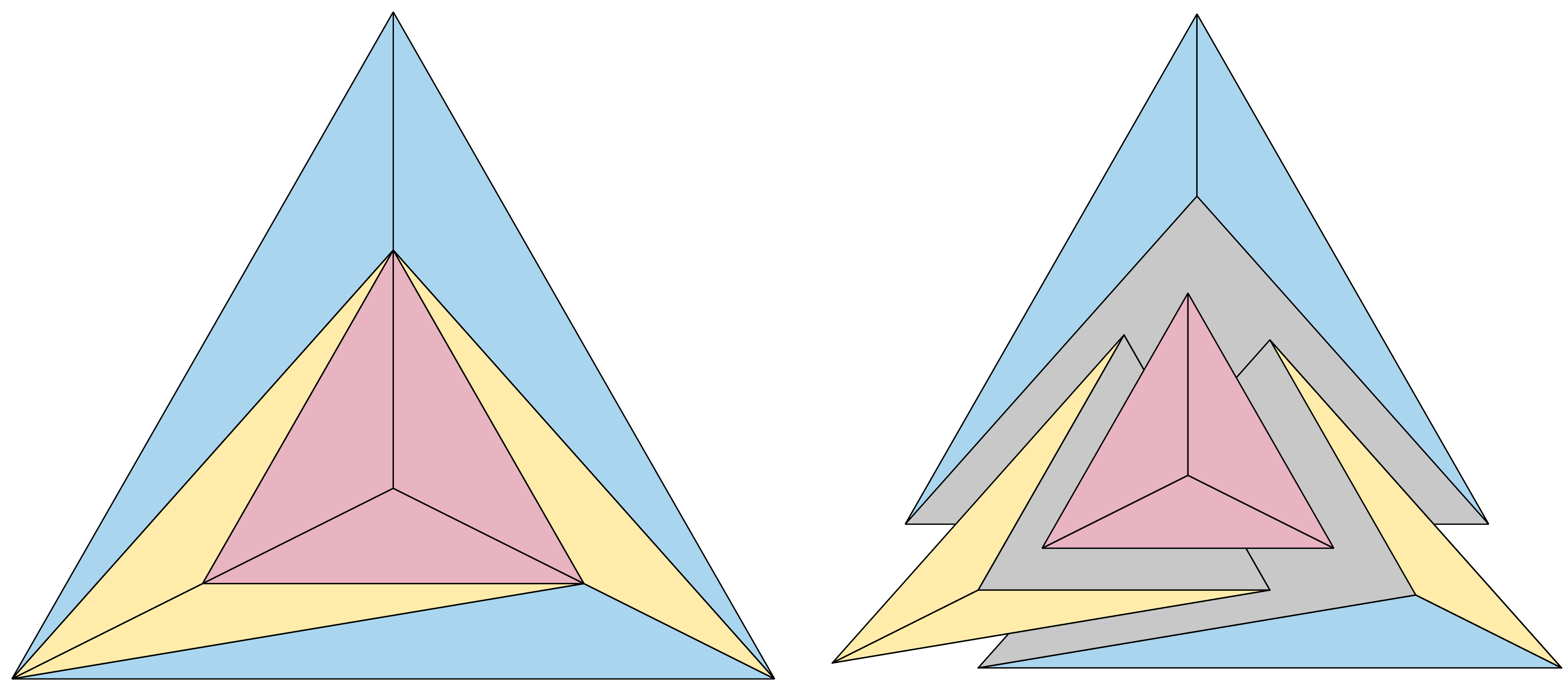}
\caption{Hat for stacked polyhedra (top view, left, and exploded view as a stacked polyhedron, right)}
\label{fig:stacked-hat}
\end{figure}

\begin{lemma}
The hat described above has no single-piece unfolding.
\end{lemma}

\begin{proof}
As with our other hat, the center vertex of this hat has positive curvature, and the other three interior vertices have negative curvature, so by \autoref{lem:hat-path} any unfolding of the hat that leaves it in one piece must form a path consisting of a single edge cutting from the boundary to the crown, two edges cutting between the band and the crown, and one edge cutting to the center of the crown.

There are many more cases than there were in \autoref{fig:hatpaths}, but we can avoid case-based reasoning by arguing that in each case, the brim and band triangles unfold in such a way that the three edges between the band and crown triangles form a polygonal chain with interior angles less than the $108^\circ$ angles of the regular pentagon (in fact,  close to $60^\circ$, because of the way we have constructed this part of the hat to differ only by a small amount from the top view shown in \autoref{fig:stacked-hat}. Therefore, just as in \autoref{fig:midline}, each edge at one end of this chain of three edges overlaps the perpendicular bisector of the edge at the other end of the chain.

Cutting along a path from a boundary edge of the hat to its center vertex forces the three crown triangles to be attached to the unfolded brim and band triangles on one of the two edges at the end of this path. However, our construction makes the three crown triangles tall enough to ensure that, no matter which of these two edges they are attached to, they will overlap the edge at the other end of the path at the point where it crosses the perpendicular bisector.\end{proof}

\begin{theorem} \label{thm:stacked-ununf}
There exists an ununfoldable stacked polyhedron.
\end{theorem}

\begin{proof}
We replace the four faces of a regular tetrahedron with the hat described above.
Each such replacement can be realized as a stacking of four tetrahedra onto the face, so the result is a stacked polyhedron.
As in \autoref{thm:ununf}, each hat lies within a prism having the
tetrahedron face as a base, so they do not overlap each other in space;
and the set of edges cut in any unfolding must include at least four paths
between the four tetrahedron vertices, necessarily forming a cycle that cuts
one part of the polyhedron surface from the rest.
\end{proof}

A stacked hat with the same combinatorial structure as the one used in this construction, with the center vertex positively curved and the surrounding three vertices negatively curved, cannot be formed from acute triangles, because that would leave the degree-four vertex with positive curvature. We leave as an open question whether it is possible for an ununfoldable stacked polyhedron to have all faces acute.

\section{Very Ununfoldable Polyhedra}

Both families of examples above
can be made into very ununfoldable families.
In both cases, the approach is the same: instead of starting from a
tetrahedron, we start from a polyhedron with many triangular faces,
and show that attaching hats to more and more triangles requires more and
more unfolded pieces.

\begin{theorem} \label{thm:many}
There exist topologically convex polyhedra with acute isosceles triangle faces
such that any unfolding formed by cutting along edges into multiple
non-self-overlapping pieces requires an unbounded number of pieces.
\end{theorem}

\begin{proof}
For any integer $k \geq 1$, refine the regular tetrahedron by subdividing each
edge into $k$ equal-length edges and subdivide each face into a regular grid of
$\sum_{i=1}^k (2 i - 1) = k^2$ equilateral triangles of side length $1/k$,
for a total of $4 k^2$ faces and (by inclusion-exclusion)
$\sum_{i=1}^{k+1} i - 6 (k+1) + 4 = 2 k^2 + 2$
vertices.
%
%
Replace each equilateral triangular face
by an acute-triangle hat pointing outward.
As in \autoref{thm:ununf}, each hat lies within a prism having the face of the
tetrahedron as a base, so they do not overlap each other in space;
and any unfolding into multiple pieces must, in each hat, either cut along a cycle within the hat or cut along some path
connecting two of its three boundary vertices. Let $c$ be the number of cycles within hats cut in this way, so that
we have a system of at least $4 k^2+c$ disjoint paths connecting pairs of
subdivided-tetrahedron vertices.

Now consider cutting the polyhedron surface along these paths, one by one.
Each cut either connects two subdivided-tetrahedron vertices that were not
previously connected along the system of cuts, or two subdivided-tetrahedron
vertices that were previously connected. If cutting along a path connects two
vertices that were not previously connected, it reduces the number of
connected components among these vertices; this case can happen at most
$2 k^2 + 1$ times.
If cutting along a path connects two vertices that were previously connected, then that
path and the path through which they were previously connected form a Jordan curve that separates off two parts of the
surface from each other. Because there are $4k^2-c$ paths connecting pairs of subdivided-tetrahedron vertices, only $2 k^2 + 1$ of which can form new connections, this case must happen at least $2 k^2 - 1-c$ times.
Because the surface started with a single piece and undergoes at least
$2 k^2 - 1-c$ separations, it ends up with at least $2 k^2-c$ pieces, which together with the $c$ additional pieces formed by cycles within hats, form a total of at least $2k^2$ pieces.
\end{proof}

\begin{theorem}
There exist topologically convex stacked polyhedra such that
any unfolding formed by cutting along edges
into multiple non-self-overlapping pieces
requires an unbounded number of pieces.
\end{theorem}

\begin{proof}
For any integer $k \geq 0$, refine the regular tetrahedron by choosing any face
and attaching to the face a very shallow tetrahedron whose apex is near the
in-center of the face, effectively splitting the face into three faces,
and repeating this process a total of $k$ times.
Because each attachment increases the number of faces by $2$ and the number
of vertices by $1$, the result is a stacked polyhedron with $4 + 2 k$
triangular faces (not necessarily equilateral) and $4 + k$ vertices.
Replace each triangle with a version of the hat from \autoref{sec:stacked}
pointed outward, using the availability flexibility to make the interface
between the band and crown an equilateral triangle
near the in-center of the original triangle.
As in \autoref{thm:stacked-ununf}, the result is a stacked polyhedron;
each hat lies within a prism having the face of the
tetrahedron as a base, so they do not overlap each other in space;
and any unfolding into multiple pieces must cut each hat
along some path connecting two of its three boundary vertices.
As in \autoref{thm:many}, at most $3 + k$ such paths can decrease the number
of connected components among the $4 + k$ vertices, leaving at least $1 + k$
paths that separate the surface into at least $2 + k$ pieces.
\end{proof}

\section*{Acknowledgments}

This research was initiated during the Virtual Workshop on Computational
Geometry organized by E. Demaine on March 20--27, 2020.
We thank the other participants of that workshop for helpful discussions
and providing an inspiring atmosphere.

\balance
\raggedright
\bibliographystyle{plainurl}
\bibliography{unfold}

\begin{thebibliography}{10}

\bibitem{AbeDem-CCCG-11}
Zachary Abel and Erik~D. Demaine.
\newblock {Edge-Unfolding Orthogonal Polyhedra is Strongly NP-Complete}.
\newblock In {\em Proceedings of the 23rd Annual Canadian Conference on
  Computational Geometry (CCCG 2011)}, 2011.
\newblock URL: \url{https://www.cccg.ca/proceedings/2011/papers/paper43.pdf}.

\bibitem{AbeDemDem-CCCG-11}
Zachary Abel, Erik~D. Demaine, and Martin~L. Demaine.
\newblock {A topologically convex vertex-ununfoldable polyhedron}.
\newblock In {\em Proceedings of the 23rd Canadian Conference on Computational
  Geometry (CCCG 2011), Toronto, August 10{--}12, 2011}, 2011.
\newblock URL: \url{https://cccg.ca/proceedings/2011/papers/paper85.pdf}.

\bibitem{AgaAroORo-SICOMP-97}
Pankaj~K. Agarwal, Boris Aronov, Joseph O'Rourke, and Catherine~A. Schevon.
\newblock {Star unfolding of a polytope with applications}.
\newblock {\em SIAM Journal on Computing}, 26(6):1689{--}1713, 1997.
\newblock \href {http://dx.doi.org/10.1137/S0097539793253371}
  {\path{doi:10.1137/S0097539793253371}}.

\bibitem{AkiDemEpp-JCDCG-19}
Hugo~A. Akitaya, Erik~D. Demaine, David Eppstein, Tomohiro Tachi, and Ryuhei
  Uehara.
\newblock {Minimal ununfoldable polyhedron}.
\newblock In {\em Abstracts from the 22nd Japan Conference on Discrete and
  Computational Geometry, Graphs, and Games (JCDCGGG 2019)}, pages 27{--}28,
  Tokyo, Japan, September 2019.
\newblock URL:
  \url{https://erikdemaine.org/papers/MinimalUnunfoldable_JCDCGGG2019/}.

\bibitem{AroORo-DCG-92}
Boris Aronov and Joseph O'Rourke.
\newblock {Nonoverlap of the star unfolding}.
\newblock {\em Discrete {\&} Computational Geometry}, 8(3):219{--}250, 1992.
\newblock \href {http://dx.doi.org/10.1007/BF02293047}
  {\path{doi:10.1007/BF02293047}}.

\bibitem{BarGho-DCG-19}
Nicholas Barvinok and Mohammad Ghomi.
\newblock {Pseudo-edge unfoldings of convex polyhedra}.
\newblock {\em Discrete {\&} Computational Geometry}, 2019.
\newblock \href {http://dx.doi.org/10.1007/s00454-019-00082-1}
  {\path{doi:10.1007/s00454-019-00082-1}}.

\bibitem{BerDemEpp-CGTA-03}
Marshall Bern, Erik~D. Demaine, David Eppstein, Eric Kuo, Andrea Mantler, and
  Jack Snoeyink.
\newblock {Ununfoldable polyhedra with convex faces}.
\newblock {\em Computational Geometry: Theory {\&} Applications},
  24(2):51{--}62, 2003.
\newblock \href {http://dx.doi.org/10.1016/S0925-7721(02)00091-3}
  {\path{doi:10.1016/S0925-7721(02)00091-3}}.

\bibitem{BieDemDem-CCCG-98}
Therese~C. Biedl, Erik~D. Demaine, Martin~L. Demaine, Anna Lubiw, Mark~H.
  Overmars, Joseph O'Rourke, Steve Robbins, and Sue Whitesides.
\newblock {Unfolding some classes of orthogonal polyhedra}.
\newblock In {\em Proceedings of the 10th Canadian Conference on Computational
  Geometry (CCCG 1998)}, 1998.
\newblock URL:
  \url{https://cgm.cs.mcgill.ca/cccg98/proceedings/cccg98-biedl-unfolding.ps.gz}.

\bibitem{CheHan-SoCG-90}
Jindong Chen and Yijie Han.
\newblock {Shortest paths on a polyhedron}.
\newblock In {\em Proceedings of the 6th Annual Symposium on Computational
  Geometry (SoCG 1990)}. ACM Press, 1990.
\newblock \href {http://dx.doi.org/10.1145/98524.98601}
  {\path{doi:10.1145/98524.98601}}.

\bibitem{Genus2Unfolding_GC}
Mirela Damian, Erik~D. Demaine, Robin Flatland, and Joseph O'Rourke.
\newblock {Unfolding genus-2 orthogonal polyhedra with linear refinement}.
\newblock {\em Graphs and Combinatorics}, 33(5):1357{--}1379, 2017.
\newblock \href {http://dx.doi.org/10.1007/s00373-017-1849-5}
  {\path{doi:10.1007/s00373-017-1849-5}}.

\bibitem{DemDemEpp-19}
Erik~D. Demaine, Martin~L. Demaine, David Eppstein, and Joseph O'Rourke.
\newblock {Some polycubes have no edge-unzipping}.
\newblock Electronic preprint arxiv:1907.08433, 2019.

\bibitem{DemDemUeh-CCCG-13}
Erik~D. Demaine, Martin~L. Demaine, and Ryuhei Uehara.
\newblock {Zipper unfoldability of domes and prismoids}.
\newblock In {\em Proceedings of the 25th Canadian Conference on Computational
  Geometry (CCCG 2013), Waterloo, Ontario, Canada August 8th{--}10th, 2013},
  2013.
\newblock URL: \url{https://cccg.ca/proceedings/2013/papers/paper_10.pdf}.

\bibitem{DemEppEri-DG-03}
Erik~D. Demaine, David Eppstein, Jeff Erickson, George~W. Hart, and Joseph
  O'Rourke.
\newblock {Vertex-unfoldings of simplicial manifolds}.
\newblock In {\em Discrete Geometry: In honor of W. Kuperberg's 60th birthday},
  volume 253 of {\em Pure and Applied Mathematics}, pages 215{--}228. Marcel
  Dekker, 2003.

\bibitem{SunUnfolding_EGC2011f}
Erik~D. Demaine and Anna Lubiw.
\newblock {A generalization of the source unfolding of convex polyhedra}.
\newblock In Alberto M{\'a}rquez, Pedro Ramos, and Jorge Urrutia, editors, {\em
  Revised Papers from the 14th Spanish Meeting on Computational Geometry},
  volume 7579 of {\em Lecture Notes in Computer Science}, pages 185{--}199,
  Alcal{\'a} de Henares, Spain, June 2011.
\newblock \href {http://dx.doi.org/10.1007/978-3-642-34191-5_18}
  {\path{doi:10.1007/978-3-642-34191-5_18}}.

\bibitem{Demaine-O'Rourke-2007}
Erik~D. Demaine and Joseph O'Rourke.
\newblock {\em {Geometric Folding Algorithms: Linkages, Origami, Polyhedra}}.
\newblock Cambridge University Press, July 2007.

\bibitem{Duerer-1525}
Albrecht D{\"u}rer.
\newblock {\em {The Painter's Manual: A Manual of Measurement of Lines, Areas,
  and Solids by Means of Compass and Ruler Assembled by Albrecht D{\"u}rer for
  the Use of All Lovers of Art with Appropriate Illustrations Arranged to be
  Printed in the Year MDXXV}}.
\newblock Abaris Books, Inc., New York, 1977.
\newblock English translation of \textit{Unterweysung der Messung mit dem
  Zirkel un Richtscheyt in Linien Ebnen und Gantzen Corporen}, 1525.

\bibitem{Fri-HFM-18}
Michael Friedman.
\newblock {\em {A History of Folding in Mathematics: Mathematizing the
  Margins}}.
\newblock Birkh{\"a}user, 2018.
\newblock See in particular page 47.
\newblock \href {http://dx.doi.org/10.1007/978-3-319-72487-4}
  {\path{doi:10.1007/978-3-319-72487-4}}.

\bibitem{Ghomi-2014}
Mohammad Ghomi.
\newblock {Affine unfoldings of convex polyhedra}.
\newblock {\em Geometry {\&} Topology}, 18:3055{--}3090, 2014.
\newblock \href {http://dx.doi.org/10.2140/gt.2014.18.3055}
  {\path{doi:10.2140/gt.2014.18.3055}}.

\bibitem{Gru-Geomb-01b}
Branko Gr{\"u}nbaum.
\newblock {A convex polyhedron which is not equifacettable}.
\newblock {\em Geombinatorics}, 10(4):165{--}171, 2001.
\newblock URL:
  \url{https://sites.math.washington.edu/~grunbaum/Nonequifacettablesphere.pdf}.

\bibitem{Gruenbaum-2002-net}
Branko Gr{\"u}nbaum.
\newblock {No-net polyhedra}.
\newblock {\em Geombinatorics}, 11:111{--}114, 2002.
\newblock URL:
  \url{https://www.math.washington.edu/~grunbaum/Nonetpolyhedra.pdf}.

\bibitem{caps}
Joseph O'Rourke.
\newblock {Edge-unfolding nearly flat convex caps}.
\newblock In Bettina Speckmann and Csaba~D. T{\'o}th, editors, {\em Proceedings
  of the 34th International Symposium on Computational Geometry (SoCG 2018)},
  volume~99 of {\em LIPIcs}, pages 64:1{--}64:14, 2018.
\newblock \href {http://dx.doi.org/10.4230/LIPIcs.SoCG.2018.64}
  {\path{doi:10.4230/LIPIcs.SoCG.2018.64}}.

\bibitem{ORo-19}
Joseph O'Rourke.
\newblock {Unfolding polyhedra}.
\newblock Electronic preprint arxiv:1908.07152, 2019.

\bibitem{She-MPCPS-75}
G.~C. Shephard.
\newblock {Convex polytopes with convex nets}.
\newblock {\em Mathematical Proceedings of the Cambridge Philosophical
  Society}, 78(3):389{--}403, 1975.
\newblock \href {http://dx.doi.org/10.1017/s0305004100051860}
  {\path{doi:10.1017/s0305004100051860}}.

\bibitem{Tar-UMN-99}
A.~S. Tarasov.
\newblock {Polyhedra that do not admit natural unfoldings}.
\newblock {\em Uspekhi Matematicheskikh Nauk}, 54(3):185{--}186, 1999.
\newblock \href {http://dx.doi.org/10.1070/rm1999v054n03ABEH000171}
  {\path{doi:10.1070/rm1999v054n03ABEH000171}}.

\bibitem{Wen-PM-71}
Magnus~J. Wenninger.
\newblock {\em {Polyhedron Models}}.
\newblock Cambridge University Press, 1971.

\end{thebibliography}

\end{document}